\documentclass[aps,pra,twocolumn,superscriptaddress,floatfix,
nofootinbib,showpacs,longbibliography]{revtex4-1}

\usepackage[utf8]{inputenc}  
\usepackage[T1]{fontenc}     %Output what you want e.g., é, ł, a, ü
\usepackage[british]{babel}  %Do hyphenation according to british english
\usepackage[sc,osf]{mathpazo}
\usepackage[scaled=0.86]{berasans}  % URL font that go well wtih palatino
\usepackage[colorlinks=true, citecolor=blue, urlcolor=blue]{hyperref}  %Hyperlinks (pink, green, blue)
\usepackage{graphicx} % Package to insert exteral figures
\usepackage[babel]{microtype}  %Improves text justification
\usepackage{amsmath,amssymb,amsthm,bm,amsfonts,mathrsfs,bbm} %Usefull math packages

\usepackage{xspace}  %Useful to add space in macros
\usepackage{pgf,tikz}
\usepackage{xcolor}
\usepackage{multirow}
\usepackage{array}
\usepackage{bigstrut}
\usepackage{braket}
\usepackage{color}
\usepackage{natbib}
\usepackage{multirow}
\usepackage{mathtools}
\usepackage{float}
\usepackage[caption = false]{subfig}
\usepackage{xcolor,colortbl}
\usepackage{color}
\newcommand{\Tr}{\operatorname{Tr}}

\newcommand{\be}{\begin{equation}}
\newcommand{\ee}{\end{equation}}
\newcommand{\ba}{\begin{eqnarray}}
\newcommand{\ea}{\end{eqnarray}}

\newtheorem{proposition}{Proposition}

\newtheorem{lemma}{Lemma}
\begin{document}

\title{Advantage of Quantum Theory Over Non-classical Models of Communication}

\author{Sutapa Saha}
%\email{xyz}    
\affiliation{Physics and Applied Mathematics Unit, Indian Statistical Institute, 203 B.T. Road, Kolkata 700108, India.}

\author{Some Sankar Bhattacharya}
%\email{xyz}    
\affiliation{Department of Computer Science, The University of Hong Kong, Pokfulam Road, Hong Kong.}

\author{Tamal Guha}
%\email{xyz}    
\affiliation{Physics and Applied Mathematics Unit, Indian Statistical Institute, 203 B.T. Road, Kolkata 700108, India.}

\author{Saronath Halder}
%\email{xyz}    
\affiliation{Quantum Information and Computation Group, Harish-Chandra Research Institute, HBNI, Chhatnag Road, Jhunsi, Prayagraj (Allahabad) 211 019, India
}

\author{Manik Banik}
%\email{xyz}    
\affiliation{School of Physics, IISER Thiruvananthapuram, Vithura, Kerala 695551, India.}

\begin{abstract}
	Quantum correlations provide dramatic advantage over the corresponding classical resources in several communication tasks. However a broad class of probabilistic theories exists that attributes greater success than quantum theory in many of these tasks by allowing supra-quantum correlations in `space-like' and/or `time-like' paradigms. In this letter we propose a communication task involving three spatially separated parties where one party (verifier) aims to verify whether the bit strings possessed by the other two parties (terminals) are equal or not. We call this task {\it authentication with limited communication}, the restrictions on communication being: (i) the terminals cannot communicate with each other, but (ii) each of them can communicate with the verifier through single use of channels with limited capacity. Manifestly, classical resources are not sufficient for perfect success of this task. Moreover, it is also not possible to perform this task with certainty in several non-classical theories although they might possess stronger `space-like' and/or `time-like' correlations. Surprisingly quantum resources can achieve the perfect winning strategy. The proposed task thus stands apart from all previously known communication tasks as it exhibits quantum advantage over other non-classical strategies.
\end{abstract}

%\pacs{03.65.Ta,03.65.Ud, 03.67.Dd}
%\keywords{}

% 03.65.Ta	Foundations of quantum mechanics;
% 03.67.Dd	Quantum cryptography and communication security
% 03.67.Hk	Quantum communication
% 03.65.Ud	Entanglement and quantum nonlocality

\maketitle
%\section{Introduction}

Advent of quantum information theory identifies useful applications of quantum mechanics over its classical counterpart in several computational as well as information theoretic protocols \cite{Shor84,Bennett84,Bennett92,Bennett93,Buhrman10,Pironio10,Colbeck12,Dale15}. 
%It admits extremely efficient algorithms, such as Shor's factoring algorithm \cite{Shor84}, advantageous communication tasks, viz. quantum teleportation \cite{Bennett93} and quantum super dense coding \cite{Bennett92}, qualitatively superior cryptographic protocols, eg. the BB-$84$ key distribution protocol \cite{Bennett84}, effective algorithms in a distributed computational task known as communication complexity \cite{Buhrman10}, and also shows provable advantage in randomness processing \cite{Pironio10,Colbeck12,Dale15}. Development of quantum error-correcting codes and fault-tolerant quantum computation \cite{Terhal15} make many of these noble concepts realizable in experiment \cite{Mattle96,Bouwmeester97,Vandersypen01,Gisin02}. 
However, in many cases, it is notoriously hard to find which particular feature(s) of quantum theory like coherent superposition, continuity of state space, non-classical correlations, viz. nonlocality/entanglement/quantum discord accounts for quantum advantage in a particular task. 

A more general mathematical modeling of an operational theory is possible under the framework of generalized probability theories (GPTs) which incorporates several non classical features of quantum theory and thus manifests many advantageous protocols \cite{Barrett05,Barnum07,Barnum08,Scarani12,Banik15(1)}. For example, in the distributed computing setting, where several spatially separated computing devices are allowed to exchange limited communications in order to perform some computational task, quantum \emph{nonlocal correlations} can provide surprising advantages \cite{Cleve97,Brukner04}. Interestingly, in such cases, one can come up with more dramatic correlations that satisfy the relativistic causality or more broadly no-signaling (NS) principle but at the same time exhibit advantage over the quantum correlations -- Popescu-Rohrlich (PR) correlation is one such celebrated example in the bipartite setting \cite{Popescu94}. Such stronger correlations exhibit weird phenomena as reflected in violation of several physical and information theoretic principles \cite{vanDam,Brassard06,Pawlowski09,Navascues09,Fritz13}.
On the other extreme, a different toy theory is also possible that contains only local correlations but can supersede quantum theory in certain communication task by allowing stronger `time-like' correlations. Such an anomalous behavior has been reported very recently by the name of \emph{hypersignaling} (HS) phenomena \cite{DallArno17}. 

Existence of such non-classical toy theories thus provoke an important question: what makes quantum theory special in operational sense? In other words, does there exist some task(s) where quantum theory outperforms these non-classical toy theories? Answer to this question is partially known from the perspective of computational power of a physical theory \cite{Aaronson04,Krumm18}. It has been shown that several beyond-quantum models of computation are trivial, i.e., the set of reversible transformations consists entirely of single-bit gates, and not even classical computation is possible \cite{Krumm18}. However it is known that the class of functions computable with classical physics exactly coincides with the class computable quantum mechanically, and the quantum exponential speed-up over classical computation for a range of problems, such as factoring, is based upon the strong believe about persistence of polynomial hierarchy \cite{Harrow17}.

In this letter we approach this question from a different outlook -- from the perspective of a communication task. Interestingly we find that there exists a communication task that can perfectly be won in quantum theory but the success probability of this task is limited not only in classical theory but also in HS model and PR model. Our task involves three spatially separated parties, where two parties are given random two-bit strings. The third party acts as a verifier who has to verify whether these strings are identical or not. The first two parties cannot communicate with each other but can encode their messages in the state of some physical system and consequently send it to the verifier. However single-shot information carrying capacity, namely the \emph{signaling dimension}, of these physical systems are limited to two. 
We call this task \emph{authentication with limited communication} (ALC). Naturally the question arises which feature of quantum theory makes it quintessential for perfectly winning the ALC task even though it allows limited correlations in space-like and time-like paradigms compared to other non-classical GPTs. At this point we note that though PR theory is more radical than quantum mechanics in allowing joint state space structure and hence stronger nonlocal correlations but it is conservative in comparison to the later one to allow measurement in entangled bases. The HS model is the other extreme: it allows more general kind of measurements than quantum theory but grants only local correlations \cite{DallArno17}. We then consider other two theories, namely Hybrid model and frozen model, lying in between PR theory and HS model. These two theories allow entangled kind of states as well as measurements in entangled bases. However we show that perfect success of ALC is not possible even in those models. This indicates that the perfect success of ALC in quantum theory depends on the more intricate structure of the theory. To apprehend this intricate nature we define the ALC task in a generic convex model of operational theories also known as GPT framework \cite{Mackey63,Ludwig1967,Ludwig1968,Mielnik1969,Hardy01,Barrett07,Hardy11,Chiribella11,Masanes11,Janotta14,Janotta13}.

%\section{Operational Model framework}
%The origin of the convex operational framework dates back to $1960$s with the  aim to investigate axiomatic derivations of the Hilbert space formalism of quantum theory from operational postulates \cite{Mackey63,Ludwig1967,Ludwig1968,Mielnik1969}. Recently the approach has gained renewed interest from researchers in quantum information theory exploring the information theoretic foundations of quantum theory \cite{Hardy01,Barrett07,Hardy11,Chiribella11,Masanes11,Janotta14}. 

%A GPT is specified by a list of system types, together with composition rules specifying which system type describes the combination of several other types.
In a GPT, each system is described by some state $\omega$ which specifies outcome probabilities for all measurements that can be performed on it. A complete representation of the state is achieved by listing the outcome probabilities for measurements belonging to ‘fiducial set’ \cite{Hardy01}. The set of possible states $\Omega$ of a given system type is a compact and convex set embedded in positive convex cone $V_+$ of some real vector space $V$. %Convexity of $\Omega$ assures that any statistical mixture of states is a valid state. The extremal points of the set $\Omega$ that do not allow any convex decomposition in terms of other states are called pure states or states of maximal knowledge.
An effect $\mathit{e}$ is a linear functional on $\Omega$ that maps each state onto a probability, i.e., $e:\Omega\mapsto[0,1]$, with  $\mathit{e}(\omega)$ bearing the interpretation of successfully filter the effect $e$ on the system state $\omega$. The set of effects $\Omega^*$ is embedded in the positive dual cone $(V^*)_+$ \cite{footnote}. The unit effect $u$ is defined as, $u(\omega)=1,~\forall~\omega\in\Omega$. A $d$-outcome measurement is specified by a collection of $d$ effects $M\equiv\{\mathit{e}_j~|~\sum_je_j=u\}$ such that $\sum_{j}\mathit{e}_j(\omega) = 1$ for all valid states $\omega$.
% A set of state $\{\omega_i\}_i$ is perfectly distinguishable in a single shot measurement if there exists some measurement $M=\{e_j\}_j$ such that $e_j(w_i)=\delta_{ij}$. 
A transformation $T$ maps states to states, i.e., $T:\Omega\mapsto\Omega$. Similarly as effects, they also have to be linear in order to preserve statistical mixtures. 
%Under a valid transformation the total probability cannot increase, but can decrease in general.

GPT framework also considers composite systems with local state spaces (say) $\Omega_1$ and $\Omega_2$. Such a composition must be constructed in accordance with NS principle that prohibits instantaneous communication between two spatially separated locations. NS along with another less intuitive assumption called tomographic locality \cite{Hardy13}, sufficiently implies that the state space of the composite system lives in the vector space $V_1\otimes V_2$\cite{Barrett07}. We denote the composite state space as $\Omega\equiv\Omega_1\otimes\Omega_2=(V_1\otimes V_2)^1_+$, where $(V_1\otimes V_2)^1_+$ denotes the normalized positive cone with normalization given by the order unit $u_1\otimes u_2\in V_1^*\otimes V^*_2$. There is no unique choice for the positive cone, but it lies within
the two extremes, $(V_1\otimes_{\mbox{min}}V_2)_+:= \{\sum \alpha_{ij}\omega^i_1\otimes\omega^j_2~|~\alpha_{ij}\in\mathbb{R}_{\ge 0}\}$ and $(V_1\otimes_{\mbox{max}}V_2)_+:=(V^*_1\otimes_{\mbox{min}}V^*_2)^*_+$.
While the local state spaces are simplexes, which is the case for classical probability theory with discrete event space, the choice of tensor product is unique \cite{Namioka1969}. The quantum mechanical tensor product is neither the minimal one nor the maximal one, lies strictly in between.

%A GPT system can be used as classical information carrier. In a generic communication protocol a sender aims to send some classical information $x$ appearing with probability $p(x)$ to a spatially separated receiver. 

In a communication protocol using such a GPT the sender encodes the classical message $x$ into some state $\omega_x$ and sends the encoded system to the receiver who decode the message by performing some measurement $M=\{e_y\in\mathcal{E}~|~\sum e_y=u\}$. Given a message $x\in X$ the probability of getting the outcome $y\in Y$ is $p(y|x):=e_y(\omega_x)$ and the mutual information $I(X:Y):=\sum_{xy}p_{xy}\log_2[p(xy)/p(x)p(y)]$ quantifies the amount of classical information transmitted through such a protocol. The Holevo capacity $\mathcal{H}(\Omega)$, for a system type with state space $\Omega$, is defined as the maximum of $I(X:Y)$, over all probability distributions $p(x)$, all encoding strategies, and all decoding measurements \cite{Holevo73}. 
%It may be possible to encode arbitrary large amount (even unbounded amount) of classical information in some GPT system. However 
%The Holevo capacity is always bounded by the \emph{signaling dimension} $\kappa(\Omega)$ defined in \cite{DallArno17}, i.e., $\mathcal{H}(\Omega)\le\log_2\kappa(\Omega)$. 
In the present letter we are interested in the single-shot capacity of a GPT channel given by signaling dimension. The operational definition of signaling dimension involves a communication scenario\cite{DallArno17}. As shown in \cite{Frenkel15}, signaling dimension of composite quantum system can not be greater than product of that of component subsystems. However, for every GPT this is not true in general and such a GPT manifests the hypersignaling phenomena \cite{DallArno17}.  

%\section{The ALC task}
The ALC task can be presented as a game involving three spatially separated players. Alice and Bob are two non-communicating players who are given random two-bit strings $x\in\{0,1\}^2$ and $y\in\{0,1\}^2$, respectively. Charlie is the verifier whose goal is to verify whether the strings given to Alice and Bob are identical or not. If there is no restriction on the amount of communications that Alice and Bob can convey to the verifier then there is no reason to not accomplishing the goal with perfect success. However the game has to be played under restricted communication scenario. Each of the players can encode their respective message in the state of some GPT system and subsequently send the system to the verifier through memoryless channels (see Fig.\ref{fig}). 
Access to channels with memory reduces the present task to the familiar dense coding protocol. Limitation to memoryless channels deems the ALC task to be weaker than dense coding.
Moreover the signaling dimension of the GPT system cannot be more than $2$. Though the players cannot communicate with each other, they are allowed to make their respective encoding systems correlated, i.e., they can use some composite state $\omega_{AB}\in\Omega_A\otimes\Omega_B$, where $\Omega_A$ and $\Omega_B$ denote Alice's and Bob's state spaces, respectively. While in classical theory, this implies that, the players can use only some classical correlation, in quantum theory they can use entangled states and in non-classical GPTs even more generic composite states can be used. For decoding, the verifier performs a two outcome measurement on the composite state space and depending on the measurement result he tries to authenticate whether $x=y$ or not.  
\begin{figure}[t!]
	\begin{center} 
		\includegraphics[scale=0.3]{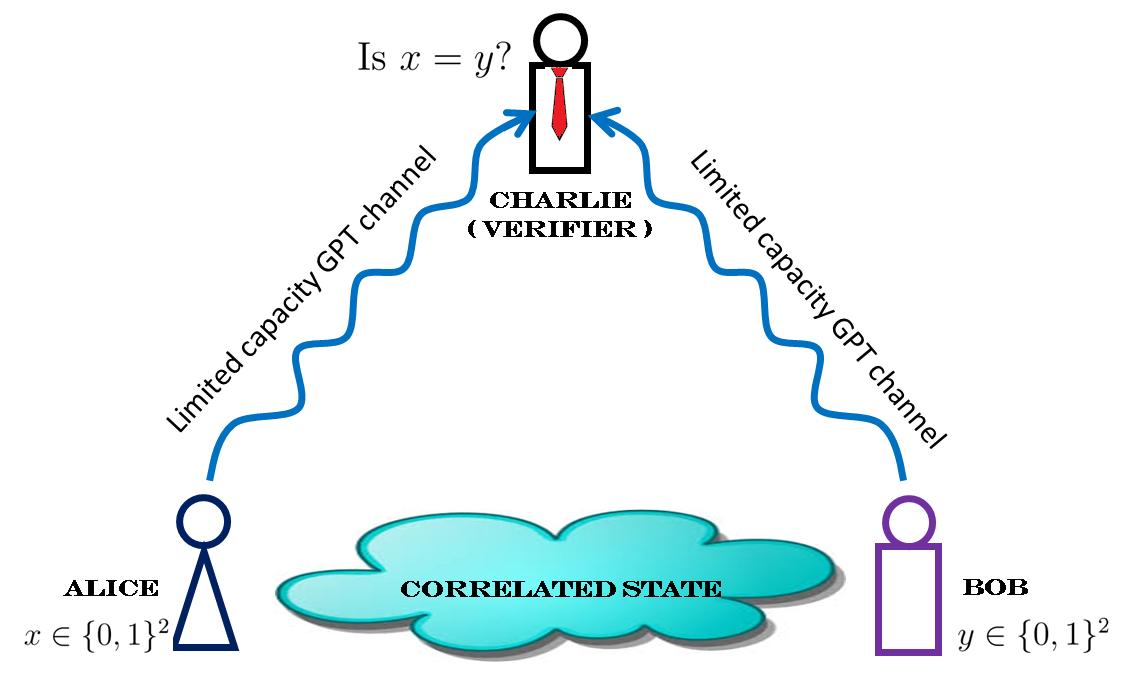}
	\end{center}
	\caption{(Color on-line) The ALC task. Alice and Bob encode their messages in some GPT state and send the encoded systems to the verifier. Each of the channels (memory less) from the players to the verifier has Holevo capacity $1$. Communication between Alice and Bob are not allowed but they can share some composite GPT state. Verifier (Charlie) performs some two outcome measurement on the composite systems received from Alice and Bob and accordingly tries to answer whether their strings are identical or not.}
	\label{fig} 
\end{figure}    

An asymptotic version of the ALC task has already been studied in quantum theory by the name quantum fingerprinting \cite{Buhrman2001}, which was originally introduced by Yao to address a particular model of communication complexity namely simultaneous message passing model \cite{Yao1979}. There Alice and Bob are given two random $n$-bit strings. Charlie has to answer whether their strings are equal or not while minimizing the amount of information that Alice and Bob send to him. We consider the simplest version of the task with $n=2$ with a prior limitation on the amount of communications. While the goal in \cite{Buhrman2001} was to establish an exponential quantum-classical gap for the equality problem in the simultaneous message passing model, here our aim is to establish quantum advantage in communication task not only over the classical theory but also over some other non-classical GPT models of communication. Subsequently, we study the ALC task in different theories. 

%\section{ALC in different GPTs}
{\it\textbf{ Classical theory:}} Classical theory arises as a special case of generalized probability theories (GPTs). State space of a classical system having signaling dimension $\kappa$ is a $(\kappa-1)$ simplex. The restriction on communication in the ALC task compels Alice and Bob to encode their messages in $1$-simplex which geometrically represents a line segment. To perform the ALC task in classical theory, the players can undergo the following naive protocol: both Alice and Bob send the first bit of their strings and the verifier answers $x=y$, if he obtains identical bits, otherwise answers $x\neq y$. The average success probability under this strategy is $3/4$. However the players can follow more general strategies-- pure, mixed or shared. A pure strategy can be defined as a tuple $(\mathrm{E}_A,\mathrm{E}_B,\mathrm{D})$, where $\mathrm{E}_A:\{0,1\}^2\mapsto\{0_A,1_A\}$ and $\mathrm{E}_B:\{0,1\}^2\mapsto\{0_B,1_B\}$ are some encoding strategies for Alice and Bob, respectively, and $\mathrm{D}:\{0_A,1_A\}\times\{0_B,1_B\}\mapsto\{0,1\}$ is some decoding strategy for the verifier. A mixed strategy is a tuple $(P_{\mathrm{E}_A},P_{\mathrm{E}_B},P_{\mathrm{D}})$, where $P_Z$ denotes distribution over $Z\in\{E_A,E_B,D\}$. They can use classical correlation to get a shared strategy $(\lambda_{\mathrm{E}_A\mathrm{E}_B},P_{\mathcal{D}})$, where $\lambda$ is shared randomness between Alice and Bob and in general $\lambda_{\mathrm{E}_A\mathrm{E}_B}\neq P_{\mathrm{E}_A}P_{\mathrm{E}_B}$. 

\begin{proposition}
	There exists no perfect classical strategy for the ALC task, neither pure, nor mixed, nor shared.
\end{proposition} 

\begin{proof}
	As already mentioned Alice and Bob can follow encoding strategies that are pure, mixed or shared. However due to convexity it will be sufficient to consider only the pure encoding strategies for finding the optimal success in ALC task. Alice's and Bob's (Charlie's) encoding (decoding) refers to a partitioning of the strings $\{00,~01,~10,~11\}$ into two disjoint sets. Such a non-trivial partitioning can be of two types- (i) one-vs-three: we call these {\it full} encoding (decoding), (ii) two-vs-two: we call these {\it partial} encoding (decoding). As the name suggests {\it partial} strategies refer to the scenario where a part of the string is ignored (it could be one of the bits or their parity). There are $20$ {\it full} and $36$ {\it partial} encoding (decoding) strategies. It is important to note that the objective is to {\it authenticate/compare} the bit strings of Alice and Bob rather than {\it knowing} the bit strings separately. While a partial encoding and decoding proves to optimal when the task is to {\it know} the bit strings\cite{Elron07}, {\it comparing} them requires {\it full} encoding and decoding. Optimizing over all such {\it full} encoding-decoding strategies The optimal classical success turns out to be $13/16$. A representative strategy which achieves this optimal values is the following: Alice's (Bob's) encoding:-- $00\mapsto 0_{A(B)};~~ 01,10,11\mapsto 1_{A(B)}$; verifier's decoding:-- $x=y$, if obtained $0_{A}$ and $0_{B}$ from Alice and Bob, otherwise $x\neq y$.
\end{proof}

{\it\textbf{ Quantum theory:}} In this case both Alice and Bob can use state of a qubit to encode their messages, i.e., the encoding state space is the set of density operator $\mathcal{D}(\mathbb{C}^2)$ acting on the Hilbert space $\mathbb{C}^2$ which is isomorphic to unit sphere in $\mathbb{R}^3$. A general encoding strategy for Alice is a mapping, $\mathrm{E}^q_A:x \mapsto \rho^x_A\in\mathcal{D}(\mathbb{C}^2_A)$ and similarly for Bob, $\mathrm{E}^q_B:y \mapsto \rho^y_B\in\mathcal{D}(\mathbb{C}^2_B)$. Verifier performs a two outcome positive operator valued measure (POVM) $M\equiv\{M_0,M_1~|~M_i>0,~i\in\{0,1\};~M_0+M_1=\mathbb{I}_4\}$ on the composite system $\mathbb{C}^2_A\otimes\mathbb{C}^2_B$ and answers $x=y$ while $M_0$ clicks, otherwise answers $x\neq y$. However for uncorrelated (product states) quantum strategies we have the following no-go result. 

\begin{lemma}
	There is no perfect quantum uncorrelated strategy for the ALC task.
\end{lemma}
\begin{proof}
	Since $x$ and $y$ are given randomly, Alice and Bob can obtain the strings in $16$ different possible ways. In $4$ cases the strings are identical and in other cases they are different. For perfect uncorrelated strategy: (i) Preparation by Alice (Bob) $\rho^x_A$ ($\sigma^y_B$) on receiving string $x$ ($y$) has to be different for all $x$ ($y$), (ii) the subspace spanned by the product states for identical strings must be orthogonal to the subspace spanned by the product states for different strings, i.e., $Tr[(\rho^x_A\otimes\sigma^y_B)M_i]=1$ for $x=y$, and $0$ otherwise, while $Tr[(\rho^x_A\otimes\rho^y_B)M_{i\oplus 1}]=1$ for $x\neq y$, and $0$ otherwise. 
	Now let's consider the preparation corresponding to $x=00$ and $y=00$. Evidently this has to belong to the subspace orthogonal to the subspace spanned by preparations corresponding to $x=00$ and $y=01,10,11$, which means $\rho^{00}_A\otimes\sigma^{00}_B$ has to be orthogonal to each of  $\{\rho^{00}_A\otimes\sigma^{01}_B,\rho^{00}_A\otimes\sigma^{10}_B,\rho^{00}_A\otimes\sigma^{11}_B\}$. But it is not possible to satisfy this requirement in $\mathbb{C}^2\otimes\mathbb{C}^2$. 
\end{proof}
Interestingly, if the players start their protocol with two-qubit entangled state then the ALC task can be perfectly won. 

\begin{proposition}
	There exists a perfect quantum entangled strategy for the ALC task.
\end{proposition}
\begin{proof}
	Let Alice and Bob share a two qubit singlet state $|\psi^-\rangle_{AB}=\frac{1}{\sqrt{2}}(|0\rangle_A\otimes|1\rangle_B-|1\rangle_A\otimes0\rangle_B)$. Consider the mapping $\{0,1\}^2\mapsto k$, with $k\in\{0,1,2,3\}$ as follows $00\mapsto 0,01\mapsto 1,10\mapsto 2,11\mapsto 3$. Whenever Alice (Bob) obtains a string $x$ ($y$) she (he) applies $\sigma_k$ ($\sigma_{k^\prime}$) on her (his) part of the singlet state and sends that part to the verifier, where $\sigma_0=\mathbb{I}$ and rest are Pauli matrices. Verifier obtains the state $\sigma_k\otimes\sigma_{k^\prime}|\psi^-\rangle_{AB}$ and performs the measurement, $M\equiv\{|\psi^-\rangle_{AB}\langle\psi^-|,\mathbb{I}-|\psi^-\rangle_{AB}\langle\psi^-|\}$. Whenever $k=k^\prime$, verifier gets the state $|\psi^-\rangle_{AB}$, otherwise he gets one of the rest three Bell states. Hence this protocol gives perfect success probability.
\end{proof}

\emph{Remark}: Note that for perfect quantum strategy both the entangled state for encoding and the measurement in entangled basis for decoding have been used. Furthermore, after the protocol verifier only knows whether Alice's and Bob's string are identical or not but no other information about the individual strings is revealed to him.  

{\it\textbf{Square bit theory:}} This particular toy model of GPT allows more generic state space structure than qubit state space.  The two dimensional state space $\mathcal{S}$ is the collection of all vectors $(x,y,1)^T\in\mathbb{R}^3$, with $-1\le x+y \le 1,~ -1\le x-y \le 1$, where $T$ denotes transposition. Shape of the state space turns out to be a square with four  pure (extremal) states,
\begin{eqnarray*}
	\omega_0&:=&(1,0,1)^T,\quad~~~\omega_1:=(0,1,1)^T,\nonumber\\
	\omega_2&:=&(-1,0,1)^T,\quad\omega_3:=(0,-1,1)^T.
\end{eqnarray*}
Specifying the outcome probability rule for the effect $e$ on state $\omega$ as $\Tr[e^T  \omega] \ge  0$, leads to the following four extremal effects,
\begin{eqnarray*}
	e_0&:=&(1,1,1)^T,\quad~~~~~~e_1:=(-1,1,1)^T,\nonumber\\
	e_2&:=&(-1,-1,1)^T,\quad e_3:=(1,-1,1)^T.
\end{eqnarray*}
The condition  $\Tr[e^T \omega] \leq 1$, $\forall~\omega$ and $\forall~ e$, implies a normalization factor $1/2$. The set of reversible  channels for the
system $\mathcal{S}$ turns out to be a finite group of
symmetries  (the dihedral  group of  order eight  $D_8$
containing   four  rotations   and  four   reflections)\cite{DallArno17}, explicitly given by,
\begin{equation*}\label{eq:single-system-unitaries}
\begin{aligned}
\mathcal{U}(\mathcal{S})=\{U_k^s: k=0,\ldots,3, s=\pm\},\\
U_k^s =
\begin{pmatrix}
\cos \frac{\pi k}2 & -s \sin \frac{\pi k}2 & 0 \\
\sin \frac{\pi k}2 & s \cos \frac{\pi k}2 & 0 \\
0 & 0 & 1
\end{pmatrix}.
\end{aligned}
\end{equation*}

State space  for composition of two such square bits $\mathcal{S}\otimes\mathcal{S}$ is a convex set in $\mathbb{R}^9$. The states $\Omega$ and the normalized effects $E$ thus can be represented by vectors in $\mathbb{R}^9$. A convenient representation can be given by $3\times  3$ real matrices rather than vectors in $\mathbb{R}^9$.  Any   bipartite  composition  naturally includes  $16$ factorized extremal states and $16$ factorized extremal effects given by,
\begin{align*}
\Omega_{4i+j} := \omega_i \otimes \omega_j^T,\qquad
E_{4i+j} := e_i \otimes e_j^T,
\end{align*}
where $i, j \in \{ 0, 1, 2, 3 \}$. One can also introduce non factorized matrices that play the role of entangled states and effects. Such an entangled state (effect) must be compatible with all factorized effects (states). Explicit calculation shows that one can have $8$ such entangled states $\{\Omega_i\}_{i=16}^{23}$ and $8$ such entangled effects $\{E_i\}_{i=16}^{23}$ that satisfy the requirement $\mbox{Tr}[E_j^T\Omega_i]\ge0$ for any $i \in [0, 15]$ and $j \in [16,  23]$, and for any $i \in  [16, 23]$ and $j\in [0, 15]$ (see Appendix \ref{App1}). While considering the bipartite theory containing entangled states and entangled effects the general consistency requirement must be fulfilled, i.e.,  all \emph{circuits} made of the allowed states and effects must give positive probabilities. For the two square-bit theory following four consistent composite models are possible \cite{DallArno17}:
\begin{enumerate}
	\item \emph{PR model:} All the  24 states $i\in[0,23]$;
	only the 16 factorized effects $j \in [0, 15]$;\vspace{-.2cm}
	\item \emph{HS  model:} Only the 16  factorized states $
	i\in[0,15]$; all the 24 effects $j \in [0, 23]$;\vspace{-.2cm}
	\item \emph{Hybrid models:}  Only 2 entangled states and
	effects  are included (along with factorized states and effects): (a) $i  \in[0,15]\cup \{20,
	22\}$   and  $j\in[0,15]\cup   \{20,22\}$;  (b) $i
	\in[0,15]\cup \{21,  23\}$ and  $j\in[0,15]\cup \{21,  23\}$;\vspace{-.2cm}
	\item \emph{Frozen Models:} Only one entangled state and
	effect is included (along with factorized states and effects), i.e.  $i \in[0,15]  \cup \{i'\}$
	and   $j\in[0,15]\cup  \{j'\}$   with  $i'   =  j'\in
	[16,23]$.
\end{enumerate}    
All these four models, like quantum theory, satisfy the no-restriction hypothesis\cite{footnote}. Since PR model consists of only factorized effects it allows more generic bipartite states and hence stronger nonlocal correlation than quantum theory resulting in violation of several principles \cite{vanDam,Brassard06,Pawlowski09,Navascues09,Fritz13}. HS model is the other extreme-- allows only factorized states and hence effects are more general than quantum. Clearly, HS model allows only local correlations and hence satisfies all bipartite principles involving space-like separated correlations. However it violates no-hypersignaling principle which imposes the restriction that signaling capacities of composite systems must be additive on signaling capacities of component subsystems \cite{DallArno17}.

While performing the ALC task in HS model, Alice and Bob can follow some product state encoding whereas the verifier has more freedom to choose the decoding measurement. On the other hand in PR model Alice and Bob have more freedom for the encoding strategy while the verifier's decoding strategy is restricted. However there is no perfect strategy in any of these two models. At this point it seems that Hybrid model and Frozen model may provide perfect success for ALC task as they allow entangled states as well as entangled effects. However we find that even in these two models it is not possible to win the ALC game with perfect success, which leads us to the following proposition (see Appendix \ref{App1} for the proof).

\begin{proposition}
	There exist no perfect strategy for the ALC task in HS model, in PR model, in Frozen model, and in Hybrid model.
\end{proposition}
An interesting point to note is that the optimal probability of success in ALC task for these models turn out to be $\frac{13}{16}$ which is same as the optimal success in classical theory. While investigating generalized NS correlations, a number of games have been studied where supra-quantum NS correlations outperform optimal quantum winning strategies \cite{Ambainis02,Pawlowski09,Oppenheim10,Brunner14,Banik15,Roy16,Ambainis16,Banik17}.There also exist games where quantum resources are as good as generalized NS correlations \cite{Brassard05,Cleve12,Arkhipov12}. On the other hand, it has been shown that even a generalized probabilistic local model can outperform quantum theory by allowing stronger `time-like' correlation \cite{DallArno17}. The ALC task, proposed in this letter, is a notable exception from all theses games: it can be won perfectly in quantum theory while several non-classical models having stronger 'space-like' and 'time-like' correlations do not provide a perfect strategy.         

PR model stands as a testimony that presence of nonlocal correlation (and hence presence of steerable/entangled state) is not a sufficient requirement to win the ALC game in a GPT. Naturally the question arises: is nonlocality\cite{Bell64, Brunner14} necessary for perfect winning the ALC task? Interestingly the answer is negative. We find that one can perfectly win the ALC task in Spekkens' toy-bit theory which is a local theory by construction \cite{Spekkens07}. The winning protocol in toy-bit theory is analogous to the quantum entangled protocol (see Appendix \ref{App2}). However the toy-bit theory is not a perfect GPT in true sense as it only allows some particular convex mixtures as valid states. Furthermore the elementary system of toy-bit theory does not satisfy the no-restriction hypothesis \cite{Janotta13}. 

On the other hand, from the example of Hybrid and Frozen models it is evident that even simultaneous presence of entangled states and entangled effects is not enough for perfect success of ALC in a GPT. One possible reason may be that the reversible dynamics in those theories are too restricted. 

At this point the question remains whether entangled/steerable states are necessary for perfect winning strategy of ALC task in an arbitrary GPT model. Answer to this question is not very obvious in general as the example of HS model suggests that information carrying capacity of composite systems can be super-additive. Therefore a more pragmatic question is whether entangled/steerable states are necessary for perfectly winning the ALC task in GPTs that respect the no-hypersignaling principle. This possibility requires further investigation. Another interesting research direction is to look for GPTs other than quantum theory that can achieve perfect success in ALC task while satisfying no-restriction hypothesis\cite{Janotta13}.     

\begin{acknowledgments}
%\emph{Acknowledgment:} 
We thank Guruprsad Kar for many stimulating discussions. We also thank Michele Dall’Arno for useful suggestions and clarifying typos in their work (private communication) \cite{DallArno17}. M.B. acknowledges the research grant of INSPIRE-faculty scheme [DST/INSPIRE/04/2017/002288] from 
Department Of Science \& Technology, Government of India. SSB is supported by the National Natural Science Foundation of China through
grant 11675136, the Foundational Questions Institute through grant FQXiRFP3-1325, the Hong Kong Research Grant Council through grant 17300918, and the John Templeton Foundation through grant 60609, Quantum Causal Structures. This publication was made possible through the support of the ID 61466 grant from the John Templeton Foundation, as part of the “The Quantum Information Structure of Spacetime (QISS)” Project (qiss.fr). The opinions expressed in this publication are those of the authors and do not necessarily reflect the views of the John Templeton Foundation.
\end{acknowledgments}
\begin{widetext}
\begin{appendix}
\section{ALC in Square-bit theory}\label{App1}
{\bf Elementary system:} The normalized state space $\mathcal{S}$ of the elementary system takes the shape of a square with the following four extremal states: 
\begin{equation}
\omega_{0}=(1,0,1)^T,~~\omega_{1}=(0,1,1)^T,~~\omega_{2}=(-1,0,1)^T,~~\omega_{3}=(0,-1,1)^T.
\end{equation}
Specifying the outcome-rule for an effect $e$ on the state $\omega$ as Tr$[{e}^T\omega]$ results in the following four normalized extremal effects:
\begin{equation}
e_{0}=\frac{1}{2}(1,1,1)^T,~~e_{1}=\frac{1}{2}(-1,1,1)^T,~~e_{2}=\frac{1}{2}(-1,-1,1)^T,~~e_{3}=\frac{1}{2}(1,-1,1)^T.
\end{equation}
The unit effect $u$ that gives Tr$[u^T\omega_i]=1$ for any $\omega_i$ is the vector $u=(0,0,1)^T$. This particular model of GPT satisfies the no-restriction hypothesis as it allows all vectors dual to the state space as valid effects. 
The set of reversible transformations $\mathcal{U}$ is the Dihedral group of order eight, given by: 
\begin{equation}
\mathcal{U}(\mathcal{S})=\{U_k^s:k=0,1,2,3,s=\pm1\},
\end{equation}
\begin{center}
	where
	$U_k^s=$
	\(\begin{pmatrix}
	$cos$\frac{\pi k}{2}& -s$ sin$\frac{\pi k}{2}& 0\\
	$sin$\frac{\pi k}{2}& s$ cos$\frac{\pi k}{2}& 0\\
	0 &0 & 1\\
	\end{pmatrix}\).
\end{center}
{\bf Bipartite composite system:} The states and effects corresponding to a composition of two elementary systems can be represented by $3\times3$ real matrices. Any bipartite composition should include all the factorized extremal states and factorized extremal effects given by:
\begin{equation}
\Omega_{4i+j}:=\omega_{i}\otimes\omega^T_{j},~~E_{4i+j}:=2~\left({e_i}\otimes{e}^T_j\right),~~i,j\in\{0,1,2,3\}.
\end{equation}
Here the factor $2$ is for normalization. One can introduce other matrices that play the role of entangled states and entangled effects. Any such entangled state (effect) must give positive probability over all factorized effects (states) according to the rule $\mbox{Tr}[E^T\Omega]$. The set of consistent normalized entangled states are given by, 
\begin{subequations}
	\begin{align}
	\Omega_{16} &=\frac{1}{2}\left(\omega_{1}\otimes \omega_{1}^T-\omega_{2}\otimes \omega_{2}^T+\omega_{2}\otimes \omega_{3}^T+\omega_{3}\otimes \omega_{2}^T\right),\\
	\Omega_{17} &=\frac{1}{2}\left(\omega_{0}\otimes \omega_{3}^T-\omega_{0}\otimes \omega_{0}^T+\omega_{1}\otimes \omega_{1}^T+\omega_{3}\otimes \omega_{0}^T\right),\\
	\Omega_{18} &=\frac{1}{2}\left(\omega_{0}\otimes \omega_{0}^T-\omega_{1}\otimes \omega_{1}^T+\omega_{1}\otimes \omega_{2}^T+\omega_{2}\otimes \omega_{1}^T\right),\\
	\Omega_{19} &=\frac{1}{2}\left(\omega_{0}\otimes \omega_{0}^T-\omega_{0}\otimes \omega_{3}^T+\omega_{1}\otimes \omega_{3}^T+\omega_{3}\otimes \omega_{2}^T\right),\\
	\Omega_{20} &=\frac{1}{2}\left(\omega_{0}\otimes \omega_{3}^T-\omega_{0}\otimes \omega_{0}^T+\omega_{1}\otimes \omega_{0}^T+\omega_{3}\otimes \omega_{1}^T\right),\\
	\Omega_{21} &=\frac{1}{2}\left(\omega_{0}\otimes \omega_{0}^T-\omega_{0}\otimes \omega_{1}^T+\omega_{1}\otimes \omega_{1}^T+\omega_{3}\otimes \omega_{2}^T\right),\\
	\Omega_{22} &=\frac{1}{2}\left(\omega_{1}\otimes \omega_{1}^T-\omega_{2}\otimes \omega_{1}^T+\omega_{2}\otimes \omega_{2}^T+\omega_{3}\otimes \omega_{0}^T\right),\\
	\Omega_{23} &=\frac{1}{2}\left(\omega_{0}\otimes \omega_{1}^T-\omega_{1}\otimes \omega_{1}^T+\omega_{1}\otimes \omega_{2}^T+\omega_{2}\otimes \omega_{0}^T\right),
	\end{align}
\end{subequations}
and the set of consistent normalized entangled effects are given by,
\begin{subequations}
	\begin{align}
	E_{16} &=\left(e_{0}\otimes e_{0}^T-e_{0}\otimes e_{3}^T+e_{1}\otimes e_{3}^T+e_{3}\otimes e_{2}^T\right),\\
	E_{17} &=\left(e_{1}\otimes e_{1}^T-e_{2}\otimes e_{2}^T+e_{2}\otimes e_{3}^T+e_{3}\otimes e_{2}^T\right),\\
	E_{18} &=\left(e_{0}\otimes e_{3}^T-e_{0}\otimes e_{0}^T+e_{1}\otimes e_{1}^T+e_{3}\otimes e_{0}^T\right),\\
	E_{19} &=\left(e_{0}\otimes e_{0}^T-e_{1}\otimes e_{1}^T+e_{1}\otimes e_{2}^T+e_{2}\otimes e_{1}^T\right),\\
	E_{20} &=\left(e_{0}\otimes e_{1}^T-e_{1}\otimes e_{1}^T+e_{1}\otimes e_{2}^T+e_{2}\otimes e_{0}^T\right),\\
	E_{21} &=\left(e_{1}\otimes e_{1}^T-e_{2}\otimes e_{1}^T+e_{2}\otimes e_{2}^T+e_{3}\otimes e_{0}^T\right),\\
	E_{22} &=\left(e_{0}\otimes e_{0}^T-e_{0}\otimes e_{1}^T+e_{1}\otimes e_{1}^T+e_{3}\otimes e_{2}^T\right),\\
	E_{23} &=\left(e_{0}\otimes e_{3}^T-e_{0}\otimes e_{0}^T+e_{1}\otimes e_{0}^T+e_{3}\otimes e_{1}^T\right).
	\end{align}
\end{subequations}
\begin{table}[b!]
	\centering
	\label{sq}
	\begin{tabular}{|c | c | c | c |}
		\hline 
		Model &  States &  Effects &  Transformations \\
		\hline\hline 
		PR Model & $\Omega_i,~i\in [0,23]$  &  $E_j,~j\in[0,15]$& $\mathcal{U}(\mathcal{S}\otimes \mathcal{S})$ \\
		\hline 
		HS Model & $\Omega_i,~i\in [0,15]$ & $E_j,~j\in[0,23]$ &  $\mathcal{U}(\mathcal{S}\otimes \mathcal{S})$  \\
		\hline 
		\multirow{2}{*}{Hybrid Model} & $\Omega_i,~i\in [0,15]\cup \{20,22\}$ &$E_j,~j\in[0,15]\cup \{20,22\}$  &  \multirow{2}{*}{$\left\{U_k^+\otimes U_l^+~|~k,l=0,2\right\}$}  \\
		& $\Omega_i,~i\in [0,15]\cup \{21,23\}$ & $E_j,~j\in[0,15]\cup \{21,23\}$ & \\
		\hline 
		\multirow{2}{*}{Frozen Model} & \multirow{2}{*}{$\Omega_i,~i\in [0,15]\cup \{n\};~n\in[16,23]$} & \multirow{2}{*}{$E_j,~j\in[0,15]\cup \{n\};~n\in[16,23]$} & $\left\{W^{0,1}\left(U_0^+\otimes U_0^+\right)\right\}, ~if~n\in [16,19]$ \\
		& & &$\left\{W^{0}\left(U_0^+\otimes U_0^+\right)\right\}, ~if~n\in [20,23]$ \\
		\hline
	\end{tabular}
	\caption{Four possible bipartite models in square-bit theory. $\mathcal{U}(\mathcal{S}\otimes\mathcal{S}):= \left\{W^{i}\left(U_j^{s_1}\otimes U_k^{s_2}\right)\right\}$ with $i\in\{0,1\};~j,k\in \{0,1,2,3\};~s_1,s_2\in \{\pm\}$, $W$ being the {\emph SWAP} map.}
\end{table}
The unit effect on the composite system is $u\otimes u^T$. Note that all the entangled effects on all the entangled states do not give rise to valid probabilities (see Table-III). As shown in \cite{DallArno17}, four consistent bipartite models are possible listed in Table-II.

\begin{table}[t!]
	\centering
	\label{tab-squit-all-prob}
	\begin{tabular}{| c |c | c | c | c |c | c | c | c |c | c | c | c |c | c | c | c |c || c | c | c |c | c | c | c |c|}
		\hline
		&\multicolumn{17}{|c|}{{\bf Factorized Effects}}&\multicolumn{8}{|c|}{{\bf Entangled Effects}}\\
		\hline
		&&$E_0$&$E_1$&$E_2$&$E_3$&$E_4$&$E_5$&$E_6$&$E_7$&$E_8$&$E_9$&$E_{10}$&$E_{11}$&$E_{12}$&$E_{13}$&$E_{14}$&$E_{15}$&$E_{16}$&$E_{17}$&$E_{18}$&$E_{19}$&$E_{20}$&$E_{21}$&$E_{22}$&$E_{23}$\\
		\hline\hline
		\parbox[t]{3mm}{\multirow{16}{*}{\rotatebox[origin=c]{90}{{\bf Factorized States}}}}&$\Omega_{0}$&$1$&$0$&$0$&$1$&$0$&$0$&$0$&$0$&$0$&$0$&$0$&$0$&$1$&$0$&$0$&$1$&$0$&$0$&$1$&$1$&$0$&$1$&$1$&$0$\\
		\cline{2-26}
		&$\Omega_{1}$&$1$&$1$&$0$&$0$&$0$&$0$&$0$&$0$&$0$&$0$&$0$&$0$&$1$&$1$&$0$&$0$&$1$&$0$&$0$&$1$&$0$&$1$&$1$&$0$\\
		\cline{2-26}
		&$\Omega_{2}$&$0$&$1$&$1$&$0$&$0$&$0$&$0$&$0$&$0$&$0$&$0$&$0$&$1$&$1$&$1$&$0$&$1$&$1$&$0$&$0$&$1$&$0$&$0$&$1$\\
		\cline{2-26}
		&$\Omega_{3}$&$0$&$0$&$1$&$1$&$0$&$0$&$0$&$0$&$0$&$0$&$0$&$0$&$0$&$0$&$1$&$1$&$0$&$1$&$1$&$0$&$0$&$1$&$1$&$0$\\
		\cline{2-26}
		&$\Omega_{4}$&$1$&$0$&$0$&$1$&$1$&$0$&$0$&$1$&$0$&$0$&$0$&$0$&$0$&$0$&$0$&$0$&$1$&$0$&$0$&$1$&$0$&$0$&$1$&$1$\\
		\cline{2-26}
		&$\Omega_{5}$&$1$&$1$&$0$&$0$&$1$&$1$&$0$&$0$&$0$&$0$&$0$&$0$&$0$&$0$&$0$&$0$&$1$&$1$&$0$&$0$&$0$&$1$&$1$&$0$\\
		\cline{2-26}
		&$\Omega_{6}$&$0$&$1$&$1$&$0$&$0$&$1$&$1$&$0$&$0$&$0$&$0$&$0$&$0$&$0$&$0$&$0$&$0$&$1$&$1$&$0$&$1$&$1$&$0$&$0$\\
		\cline{2-26}
		&$\Omega_{7}$&$0$&$0$&$1$&$1$&$0$&$0$&$1$&$1$&$0$&$0$&$0$&$0$&$0$&$0$&$0$&$0$ &$0$&$0$&$1$&$1$&$1$&$0$&$0$&$1$\\
		\cline{2-26}
		&$\Omega_{8}$&$0$&$0$&$0$&$0$&$1$&$0$&$0$&$1$&$1$&$0$&$0$&$1$&$0$&$0$&$0$&$0$ &$1$&$1$&$0$&$0$&$1$&$0$&$0$&$1$\\
		\cline{2-26}
		&$\Omega_{9}$&$0$&$0$&$0$&$0$&$1$&$1$&$0$&$0$&$1$&$1$&$0$&$0$&$0$&$0$&$0$&$0$ &$0$&$1$&$1$&$0$&$0$&$0$&$1$&$1$\\
		\cline{2-26}
		&$\Omega_{10}$&$0$&$0$&$0$&$0$&$0$&$1$&$1$&$0$&$0$&$1$&$1$&$0$&$0$&$0$&$0$&$0$ &$0$&$0$&$1$&$1$&$0$&$1$&$1$&$0$\\
		\cline{2-26}
		&$\Omega_{11}$&$0$&$0$&$0$&$0$&$0$&$0$&$1$&$1$&$0$&$0$&$1$&$1$&$0$&$0$&$0$&$0$ &$1$&$0$&$0$&$1$&$1$&$1$&$0$&$0$\\
		\cline{2-26}
		&$\Omega_{12}$&$0$&$0$&$0$&$0$&$0$&$0$&$0$&$0$&$1$&$0$&$0$&$1$&$1$&$0$&$0$&$1$ &$0$&$1$&$1$&$0$&$1$&$1$&$0$&$0$\\
		\cline{2-26}
		&$\Omega_{13}$&$0$&$0$&$0$&$0$&$0$&$0$&$0$&$0$&$1$&$1$&$0$&$0$&$1$&$1$&$0$&$0$ &$0$&$0$&$1$&$1$&$1$&$0$&$0$&$1$\\
		\cline{2-26}
		&$\Omega_{14}$&$0$&$0$&$0$&$0$&$0$&$0$&$0$&$0$&$0$&$1$&$1$&$0$&$0$&$1$&$1$&$0$ &$1$&$0$&$0$&$1$&$0$&$0$&$1$&$1$\\
		\cline{2-26}
		&$\Omega_{15}$&$0$&$0$&$0$&$0$&$0$&$0$&$0$&$0$&$0$&$0$&$1$&$1$&$0$&$0$&$1$&$1$ &$1$&$1$&$0$&$0$&$0$&$1$&$1$&$0$\\
		\hline\hline
		\parbox[t]{3mm}{\multirow{9}{*}{\rotatebox[origin=c]{90}{{\bf Entangled States}}}}&$\Omega_{16}$&$1$&$1$&$0$&$0$&$1$&$0$&$0$&$1$&$0$&$0$&$1$&$1$&$0$&$1$&$1$&$0$ &$\cellcolor{red!35}\frac{3}{2}$&$\frac{1}{2}$&\cellcolor{red!35}$-\frac{1}{2}$&$\frac{1}{2}$&$\frac{1}{2}$&$\frac{1}{2}$&$\frac{1}{2}$&$\frac{1}{2}$\\
		\cline{2-26}
		&$\Omega_{17}$&$0$&$1$&$1$&$0$&$1$&$1$&$0$&$0$&$1$&$0$&$0$&$1$&$0$&$0$&$1$&$1$ &$\frac{1}{2}$&\cellcolor{red!35}$\frac{3}{2}$&$\frac{1}{2}$&\cellcolor{red!35}$-\frac{1}{2}$&$\frac{1}{2}$&$\frac{1}{2}$&$\frac{1}{2}$&$\frac{1}{2}$\\
		\cline{2-26}
		&$\Omega_{18}$&$0$&$0$&$1$&$1$&$0$&$1$&$1$&$0$&$1$&$1$&$0$&$0$&$1$&$0$&$0$&$1$ &\cellcolor{red!35}$-\frac{1}{2}$&$\frac{1}{2}$&\cellcolor{red!35}$\frac{3}{2}$&$\frac{1}{2}$&$\frac{1}{2}$&$\frac{1}{2}$&$\frac{1}{2}$&$\frac{1}{2}$\\
		\cline{2-26}
		&$\Omega_{19}$&$1$&$0$&$0$&$1$&$0$&$0$&$1$&$1$&$0$&$1$&$1$&$0$&$1$&$1$&$0$&$0$ &$\frac{1}{2}$&\cellcolor{red!35}$-\frac{1}{2}$&$\frac{1}{2}$&\cellcolor{red!35}$\frac{3}{2}$&$\frac{1}{2}$&$\frac{1}{2}$&$\frac{1}{2}$&$\frac{1}{2}$\\
		\cline{2-26}
		&$\Omega_{20}$&$0$&$0$&$1$&$1$&$1$&$0$&$0$&$1$&$1$&$1$&$0$&$0$&$0$&$1$&$1$&$0$ &$\frac{1}{2}$&$\frac{1}{2}$&$\frac{1}{2}$&$\frac{1}{2}$&$\frac{1}{2}$&\cellcolor{red!35}$-\frac{1}{2}$&$\frac{1}{2}$&\cellcolor{red!35}$\frac{3}{2}$\\
		\cline{2-26}
		&$\Omega_{21}$&$1$&$0$&$0$&$1$&$1$&$1$&$0$&$0$&$0$&$1$&$1$&$0$&$0$&$0$&$1$&$1$ &$\frac{1}{2}$&$\frac{1}{2}$&$\frac{1}{2}$&$\frac{1}{2}$&\cellcolor{red!35}$-\frac{1}{2}$&$\frac{1}{2}$&\cellcolor{red!35}$\frac{3}{2}$&$\frac{1}{2}$\\
		\cline{2-26}
		&$\Omega_{22}$&$1$&$1$&$0$&$0$&$0$&$1$&$1$&$0$&$0$&$0$&$1$&$1$&$1$&$0$&$0$&$1$ &$\frac{1}{2}$&$\frac{1}{2}$&$\frac{1}{2}$&$\frac{1}{2}$&$\frac{1}{2}$&\cellcolor{red!35}$\frac{3}{2}$&$\frac{1}{2}$&\cellcolor{red!35}$-\frac{1}{2}$\\
		\cline{2-26}
		&$\Omega_{23}$&$0$&$1$&$1$&$0$&$0$&$0$&$1$&$1$&$1$&$0$&$0$&$1$&$1$&$1$&$0$&$0$ &$\frac{1}{2}$&$\frac{1}{2}$&$\frac{1}{2}$&$\frac{1}{2}$&\cellcolor{red!35}$\frac{3}{2}$&$\frac{1}{2}$&\cellcolor{red!35}$-\frac{1}{2}$&$\frac{1}{2}$\\
		\hline
	\end{tabular}
	\caption{The values of $\mbox{Tr}[E_i^T\Omega_j]$ for normalized effects $E_i$ and normalized states $\Omega_j$ are listed here, $i,j\in[0,23]$. Note that the values in the shaded cells do not correspond to valid probability measures.}
\end{table}

{\bf No perfect strategy for ALC in Square-bit theories:}
We are now in a position to prove our main result: that there exists no perfect strategy for the ALC task in HS model, in PR model, in Hybrid model, and in Frozen model. For that we first prove the following lemma.\\\\

\begin{lemma}\label{fact-lemma}
There is no perfect strategy for ALC task while following factorized encodings and factorized decodings.
\end{lemma}
\begin{proof}
Consider the mapping $\{0,1\}^2\mapsto k$, with $k\in\{0,1,2,3\}$ as follows $00\mapsto 0,01\mapsto 1,10\mapsto 2,11\mapsto 3$. Let Alice and Bob encode their strings as $k\mapsto \omega_k$. For this encoding, while $x=y$ Charlie receives $\Omega_{0},\Omega_{5},\Omega_{10},\Omega_{15}$, otherwise he receives $\{\Omega_i~|~i\in[0,15]\setminus\{0,5,10,15\}\}$ (see Table-IV).
\begin{table}[h!]
	\centering
	\label{factorized}
	\begin{tabular}{| c |c | c | c | c |c | c | c | c |c | c | c | c |c | c | c | c |c |}
		\hline
		&\multicolumn{17}{|c|}{{\bf Factorized Effects}}\\
		\hline
		&&$E_0$&$E_1$&$E_2$&$E_3$&$E_4$&$E_5$&$E_6$&$E_7$&$E_8$&$E_9$&$E_{10}$&$E_{11}$&$E_{12}$&$E_{13}$&$E_{14}$&$E_{15}$\\
		\hline\hline
		\parbox[t]{3mm}{\multirow{16}{*}{\rotatebox[origin=c]{90}{{\bf Factorized States}}}}&\cellcolor{pink}$\Omega_{0}$&\cellcolor{pink}$1$&\cellcolor{pink}$0$&\cellcolor{pink}$0$&\cellcolor{pink}$1$&\cellcolor{pink}$0$&\cellcolor{pink}$0$&\cellcolor{pink}$0$&\cellcolor{pink}$0$&\cellcolor{pink}$0$&\cellcolor{pink}$0$&\cellcolor{pink}$0$&\cellcolor{pink}$0$&\cellcolor{pink}$1$&\cellcolor{pink}$0$&\cellcolor{pink}$0$&\cellcolor{pink}$1$\\
		\cline{2-18}
		&$\Omega_{1}$&$1$&$1$&$0$&$0$&$0$&$0$&$0$&$0$&$0$&$0$&$0$&$0$&$1$&$1$&$0$&$0$\\
		\cline{2-18}
		&$\Omega_{2}$&$0$&$1$&$1$&$0$&$0$&$0$&$0$&$0$&$0$&$0$&$0$&$0$&$1$&$1$&$0$&$0$\\
		\cline{2-18}
		&$\Omega_{3}$&$0$&$0$&$1$&$1$&$0$&$0$&$0$&$0$&$0$&$0$&$0$&$0$&$0$&$0$&$1$&$1$\\
		\cline{2-18}
		&$\Omega_{4}$&$1$&$0$&$0$&$1$&$1$&$0$&$0$&$1$&$0$&$0$&$0$&$0$&$0$&$0$&$0$&$0$\\
		\cline{2-18}
		&\cellcolor{pink}$\Omega_{5}$&\cellcolor{pink}$1$&\cellcolor{pink}$1$&\cellcolor{pink}$0$&\cellcolor{pink}$0$&\cellcolor{pink}$1$&\cellcolor{pink}$1$&\cellcolor{pink}$0$&\cellcolor{pink}$0$&\cellcolor{pink}$0$&\cellcolor{pink}$0$&\cellcolor{pink}$0$&\cellcolor{pink}$0$&\cellcolor{pink}$0$&\cellcolor{pink}$0$&\cellcolor{pink}$0$&\cellcolor{pink}$0$\\
		\cline{2-18}
		&$\Omega_{6}$&$0$&$1$&$1$&$0$&$0$&$1$&$1$&$0$&$0$&$0$&$0$&$0$&$0$&$0$&$0$&$0$\\
		\cline{2-18}
		&$\Omega_{7}$&$0$&$0$&$1$&$1$&$0$&$0$&$1$&$1$&$0$&$0$&$0$&$0$&$0$&$0$&$0$&$0$ \\
		\cline{2-18}
		&$\Omega_{8}$&$0$&$0$&$0$&$0$&$1$&$0$&$0$&$1$&$1$&$0$&$0$&$1$&$0$&$0$&$0$&$0$ \\
		\cline{2-18}
		&$\Omega_{9}$&$0$&$0$&$0$&$0$&$1$&$1$&$0$&$0$&$1$&$1$&$0$&$0$&$0$&$0$&$0$&$0$\\
		\cline{2-18}
		&\cellcolor{pink}$\Omega_{10}$&\cellcolor{pink}$0$&\cellcolor{pink}$0$&\cellcolor{pink}$0$&\cellcolor{pink}$0$&\cellcolor{pink}$0$&\cellcolor{pink}$1$&\cellcolor{pink}$1$&\cellcolor{pink}$0$&\cellcolor{pink}$0$&\cellcolor{pink}$1$&\cellcolor{pink}$1$&\cellcolor{pink}$0$&\cellcolor{pink}$0$&\cellcolor{pink}$0$&\cellcolor{pink}$0$&\cellcolor{pink}$0$\\
		\cline{2-18}
		&$\Omega_{11}$&$0$&$0$&$0$&$0$&$0$&$0$&$1$&$1$&$0$&$0$&$1$&$1$&$0$&$0$&$0$&$0$\\
		\cline{2-18}
		&$\Omega_{12}$&$0$&$0$&$0$&$0$&$0$&$0$&$0$&$0$&$1$&$0$&$0$&$1$&$1$&$0$&$0$&$1$ \\
		\cline{2-18}
		&$\Omega_{13}$&$0$&$0$&$0$&$0$&$0$&$0$&$0$&$0$&$1$&$1$&$0$&$0$&$1$&$1$&$0$&$0$ \\
		\cline{2-18}
		&$\Omega_{14}$&$0$&$0$&$0$&$0$&$0$&$0$&$0$&$0$&$0$&$1$&$1$&$0$&$0$&$1$&$1$&$0$ \\
		\cline{2-18}
		&\cellcolor{pink}$\Omega_{15}$&\cellcolor{pink}$0$&\cellcolor{pink}$0$&\cellcolor{pink}$0$&\cellcolor{pink}$0$&\cellcolor{pink}$0$&\cellcolor{pink}$0$&\cellcolor{pink}$0$&\cellcolor{pink}$0$&\cellcolor{pink}$0$&\cellcolor{pink}$0$&\cellcolor{pink}$1$&\cellcolor{pink}$1$&\cellcolor{pink}$0$&\cellcolor{pink}$0$&\cellcolor{pink}$1$&\cellcolor{pink}$1$\\
		\hline
	\end{tabular}
	\caption{Outcome probabilities of factorized normalized effects $E_i$ on factorized normalized states $\Omega_j$, $i,j\in[0,15]$. For the encoding considered in Lemma \ref{fact-lemma}, Charlie receives the states $\{\Omega_{0},\Omega_5,\Omega_{10},\Omega_{15}\}$ while $x=y$, as shown by shaded rows. However no extremal effect satisfies the requirement (\ref{fact-1}).}
\end{table}

For decoding, Charlie performs some measurement,
\begin{eqnarray}
M&\equiv&\{M_{eq},M_{neq}~|~M_{eq}=\sum_ip_iE_i,~M_{neq}=\sum_jq_jE_j,~ ~M_{eq}+M_{neq}=u\otimes u^T\},\\
&&\mbox{where},~p_i,q_j\ge 0;~E_i,E_j~\mbox{are factorized effects, with}~E_i\neq E_j.\nonumber
\end{eqnarray}
For perfect decoding we require,
\begin{subequations}\label{fact-1}
	\begin{align}
	\mbox{Tr}[M_{eq}^T\Omega_i]&=1,~~ \mbox{iff}~~ i\in\{0,5,10,15\};\\
	\mbox{Tr}[M_{neq}^T\Omega_i]&=1,~~ \mbox{iff}~~ i\in[0,15]\setminus\{0,5,10,15\}.
	\end{align}
\end{subequations}
However from Table-IV it is clear that no extremal effect satisfies Eq.(\ref{fact-1}) implying no perfect strategy. This holds true for any possible factorized encoding.
\end{proof}

We will now extend this Lemma and proof a no-go theorem for HS model as stated in the following proposition.
\begin{proposition}\label{prop-HS}
There is no perfect strategy for ALC task in HS model.
\end{proposition}
\begin{proof}
In HS model only factorized states are allowed for encoding while Charlie can perform more generalized decoding measurements as entangled effects are also allowed, i.e.,       
\begin{eqnarray}
M&\equiv&\{M_{eq},M_{neq}~|~M_{eq}=\sum_ip_iE_i,~M_{neq}=\sum_jq_jE_j,~ ~M_{eq}+M_{neq}=u\otimes u^T\},\\
&&\mbox{where},~p_i,q_j\ge 0;~E_i,E_j~\mbox{are factorized or entangled effects, with}~E_i\neq E_j.\nonumber
\end{eqnarray}
Whereas from Lemma-2 it follows that factorized effects do not satisfy requirement (\ref{fact-1}), from Table-III it is evident that even entangled effects are not good for perfect strategy.
\end{proof}

While considering PR model, the encodings are factorized as well as entangled whereas the decodings are factorized only. Since local unitaries map entangled states to entangled states and factorized states to factorized states, Lemma-2 leaves open only the entangled encoding. However, in the following proposition we prove a no-go result even for such entangled encodings.

\begin{proposition}
There is no perfect strategy for ALC task in PR model.
\end{proposition}
\begin{proof}
Let Alice and Bob share the entangled state $\Omega_{16}$. For encoding they can apply local reversible operations on their respective parts depending on the strings they receive. In this case eight local reversible actions are possible for each party. Under these operations the transformed states have been shown in Table-V.  
\begin{table}[b!]
	\centering
	\label{Unitary}
	\begin{tabular}{| c |c || c | c | c |c | c | c | c |c |}
		\hline
		&\multicolumn{9}{|c|}{{\bf Alice's action}}\\
		\hline
		&&$U^+_0$&$U^+_1$&$U^+_2$&$U^+_3$&$U^-_0$&$U^-_1$&$U^-_2$&$U^-_3$\\
		\hline\hline
		\parbox[t]{3mm}{\multirow{8}{*}{\rotatebox[origin=c]{90}{{\bf Bob's action}}}}&$U^+_0$&$\cellcolor{pink}\Omega_{16}$&$\Omega_{17}$&$\cellcolor{yellow}\Omega_{18}$&$\Omega_{19}$&$\cellcolor{yellow}\Omega_{23}$&$\Omega_{22}$&$\cellcolor{yellow}\Omega_{21}$&$\Omega_{20}$\\
		\cline{2-10}
		&$U^+_1$&$\cellcolor{yellow}\Omega_{17}$&$\Omega_{18}$&$\cellcolor{yellow}\Omega_{19}$&$\Omega_{16}$&$\cellcolor{yellow}\Omega_{20}$&$\Omega_{23}$&$\cellcolor{pink}\Omega_{22}$&$\Omega_{21}$\\
		\cline{2-10}
		&$U^+_2$&$\Omega_{18}$&$\Omega_{19}$&$\Omega_{16}$&$\Omega_{17}$&$\Omega_{21}$&$\Omega_{20}$&$\Omega_{23}$&$\Omega_{22}$\\
		\cline{2-10}
		&$U^+_3$&$\Omega_{19}$&$\Omega_{16}$&$\Omega_{17}$&$\Omega_{18}$&$\Omega_{22}$&$\Omega_{21}$&$\Omega_{20}$&$\Omega_{23}$\\
		\cline{2-10}
		&$U^-_0$&$\Omega_{20}$&$\Omega_{23}$&$\Omega_{22}$&$\Omega_{21}$&$\Omega_{17}$&$\Omega_{18}$&$\Omega_{19}$&$\Omega_{16}$\\
		\cline{2-10}
		&$U^-_1$&$\Omega_{21}$&$\Omega_{20}$&$\Omega_{23}$&$\Omega_{22}$&$\Omega_{18}$&$\Omega_{19}$&$\Omega_{16}$&$\Omega_{17}$\\
		\cline{2-10}
		&$U^-_2$&$\cellcolor{yellow}\Omega_{22}$&$\Omega_{21}$&$\cellcolor{yellow}\Omega_{20}$&$\Omega_{23}$&$\cellcolor{pink}\Omega_{19}$&$\Omega_{16}$&$\cellcolor{yellow}\Omega_{17}$&$\Omega_{18}$\\
		\cline{2-10}
		&$U^-_3$&$\cellcolor{yellow}\Omega_{23}$&$\Omega_{22}$&$\cellcolor{pink}\Omega_{21}$&$\Omega_{20}$&$\cellcolor{yellow}\Omega_{16}$&$\Omega_{17}$&$\cellcolor{yellow}\Omega_{18}$&$\Omega_{19}$\\
		\hline
	\end{tabular}
	\caption{Transformation of the state $\Omega_{16}$ under local reversible actions $U_{i}^{s_1}\otimes U_{j}^{s_2}[\Omega_{16}]:=U_{i}^{s_1}\Omega_{16}\left(U_{j}^{s_2}\right)^T$, where $i,j\in\{0,1,2,3\}$ and $s_1,s_2\in\{\pm\}$. Note that under the actions $U_{i}^{s_1}\otimes U_{j}^{s_2}$ with $s_1=s_2$ the state belongs in the group $[16,\cdots,19]$ and when $s_1\neq s_2$ it belongs in $[20,\cdots,23]$. Actually this fact is a generic feature: for $s_1=s_2$ the groups $\mathcal{G}_1\equiv\{\Omega_{16},\Omega_{17},\Omega_{18},\Omega_{19}\}$ and $\mathcal{G}_2\equiv\{\Omega_{20},\Omega_{21},\Omega_{22},\Omega_{23}\}$ are closed while for $s_1\neq s_2$,  $\mathcal{G}_1\leftrightarrow\mathcal{G}_2$. For the particular encoding strategy considered below, the encoded states have been shown by pink whenever $x=y$ and by yellow when $x\ne y$.}
\end{table}

Now consider an encoding as follows: for Alice $00\mapsto U^+_0$, $01\mapsto U^-_0$,$10\mapsto U^+_2$,$11\mapsto U^-_2$; for Bob $00\mapsto U^+_0$, $01\mapsto U^-_2$,$10\mapsto U^-_3$,$11\mapsto U^+_1$. The encoded states have been shown in Table-V by pink whenever $x=y$ and by yellow when $x\ne y$. While decoding, Charlie needs to perform a two-outcome measurement such that the effect corresponding to $x=y$ clicks only on the pink colored states and the other effect only on the yellow colored states. However such a measurement is not possible since in multiple cases the same states have been assigned two different colors. For example, under the actions $U^+_0\otimes U^+_0$ (corresponds to $x=y$) and $U^-_0\otimes U^-_3$ (corresponds to $x\ne y$) Charlie obtains the same encoded state $\Omega_{16}$. Considering other encodings it also turns out to be the same. 
\end{proof}
\begin{proposition}
There is no perfect strategy for ALC task in Hybrid model.	
\end{proposition}
\begin{proof}
Consider the Hybrid models with extremal states $\Omega_i,~i\in [0,15]\cup \{20,22\}$  and extremal effects $E_j,~j\in[0,15]\cup \{20,22\}$. According to Proposition \ref{prop-HS} factorized encodings will not give the perfect success. So, Alice and Bob can start their protocol with one of the entangled states (say) $\Omega_{20}$. This model allows only two local reversible operations $\{U_0^+,U_2^+\}$ on each side. Therefore they cannot encode their four different strings reliably using two such operations and hence no perfect strategy is possible even using entangled encodings and entangled decodings. Similar argument holds true for the other Hybrid model.  
\end{proof}	

\begin{proposition}
There is no perfect strategy for ALC task in Frozen model.	
\end{proposition}
\begin{proof}
Consider the frozen model with extremal states $\Omega_i,~i\in [0,15]\cup \{16\}$ and extremal effects $E_j,~j\in[0,15]\cup \{16\}$. In this case also factorized encodings are not good (Proposition \ref{prop-HS}). On the other hand, strategy that starts with sharing entangled state is trivial in this case as this model allows only one such state. Similar reasoning also holds true for other Frozen models.   
\end{proof}	

\section{ALC in Spekkens' toy-bit model}\label{App2}
This particular toy theory is based on a principle, namely \emph{knowledge balance principle} (KBP), according to which in a state of maximal knowledge the amount of knowledge one possesses about the ontic state of the system must equal the amount of knowledge she/he lacks \cite{Spekkens07}.

{\bf Elementary system:} For an elementary system the number of questions in the canonical set is two, and consequently the number of \emph{ontic} states is four. Denote the four ontic states
as $`1$',$`2$',$`3$', and $`4$'. A pure \emph{epistemic} state is a probability distribution $\{\vec{p}=(p_1,p_2,p_3,p_4)^T|~p_i\in\{0,1/2\},~\&~\sum_{i=1}^{4}p_i=1\}$, over the ontic states. Mixed epistemic states an be obtained by considering convex mixing of pure states as defined in \cite{Spekkens07}. In accordance with KBP, there exist six pure \emph{epistemic} states (state with maximal knowledge) for an elementary system that are given by,
\begin{subequations}
	\begin{align}
	1\vee 2\equiv \left(\frac{1}{2},\frac{1}{2},0,0\right)^T,\quad
	3\vee 4\equiv \left(0,0,\frac{1}{2},\frac{1}{2}\right)^T,\quad
	1\vee 3\equiv \left(\frac{1}{2},0,\frac{1}{2},0\right)^T,\\
	2\vee 4\equiv \left(0,\frac{1}{2},0,\frac{1}{2}\right)^T,\quad
	1\vee 4\equiv \left(\frac{1}{2},0,0,\frac{1}{2}\right)^T,\quad
	2\vee 3\equiv \left(0,\frac{1}{2},\frac{1}{2},0\right)^T.
	\end{align}
\end{subequations}
Here the symbol $`\vee$' means disjunction which reads as `or'. These epistemic states can be viewed as in Fig.\ref{epis-sin}   
\begin{figure}[h!]
	\subfloat[$1\vee 2$]{\includegraphics[width = 2cm]{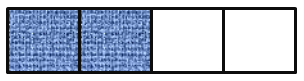}}\hspace{.5cm} 
	\subfloat[$3\vee 4$]{\includegraphics[width = 2cm]{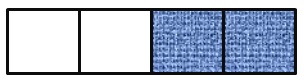}}\hspace{.5cm}
	\subfloat[$1\vee 3$]{\includegraphics[width = 2cm]{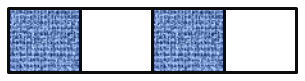}}\hspace{.5cm}
	\subfloat[$2\vee 4$]{\includegraphics[width = 2cm]{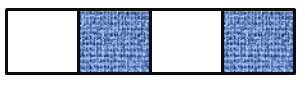}}\hspace{.5cm}
	\subfloat[$	1\vee 4$]{\includegraphics[width = 2cm]{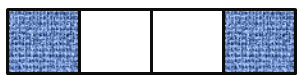}}\hspace{.5cm}
	\subfloat[$2\vee 3$]{\includegraphics[width = 2cm]{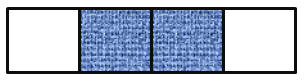}} 
	\caption{Each box denotes an ontic state, ranging $`1'$ to $`4'$ from left to right. Epistemic states are distributions on the ontic states. For example $1\vee 2$ denotes distribution on ontic state $`1'$ and $`2'$, i.e., on the first two boxes from left as shown in (a).}
	\label{epis-sin}
\end{figure}

One can introduce several quantum like features, viz., convex combination, coherent superposition, in this toy theory. State transformation as well as measurement rule are defined in accordance with KBP. Transformations are given by permutations of the ontic states and can be represented as cycles. For example, the cycle $(a)(bcd)$ means $a\mapsto a$ and $b\mapsto c\mapsto d\mapsto b$. While $4$-element permutation group contains 24 elements, only a few of them are compatible with KBP. Four allowed transformations that will be relevant for our purpose are,
\begin{subequations}\label{toy-unitary}
	\begin{align}
	U_0&=(1)(2)(3)(4),~~~~~~U_1=(12)(34),\\ 
	U_2&=(13)(24),~~~~~~~~~~~U_3=(14)(23).
	\end{align}
\end{subequations} 

{\bf Pairs of elementary systems:} In this case number of ontic states are $16$ that are represented by $a.b$, with $a,b\in\{1,2,3,4\}$. Pure epistemic states are of two type:\\
$~~~~~~~$Type-1: $(a\vee b).(c\vee d)\equiv (a.c)\vee (a.d)\vee (b.c)\vee (b.d)$; where $a,b,c,d\in{1,2,3,4}$ and $a\ne b,~c\ne d$.\\
$~~~~~~~$Type-2: $(a.e)\vee (b.f)\vee (c.g)\vee (d.f)$; where $a,b,c,d,e,f,g,h\in\{1,2,3,4\}$, $a\neq b\neq c\neq d$ and $e\neq f\neq g\neq h$.\\
While first type corresponds to factorized states, the later one corresponds to entangled states. Four such entangled states relevant to our purpose are,
\begin{subequations}
	\begin{align}
	\psi_0:=(1.1)\vee(2.2)\vee(3.3)\vee(4.4),\\
	\psi_1:=(1.2)\vee(2.1)\vee(3.4)\vee(4.3),\\
	\psi_2:=(1.3)\vee(2.4)\vee(3.1)\vee(4.2),\\
	\psi_3:=(1.4)\vee(2.3)\vee(3.2)\vee(4.1).
	\end{align}
\end{subequations}
These four states are analogous to the four Bell states of two-qubit quantum system and pictorially they can be represented as in Fig.\ref{epis-bi}.

\begin{figure}[h!]
	\subfloat[$\psi_0$]{\includegraphics[width = 2cm]{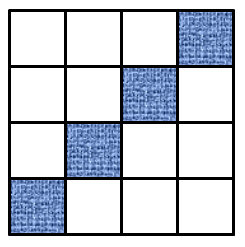}}\hspace{1cm} 
	\subfloat[$\psi_1$]{\includegraphics[width = 2cm]{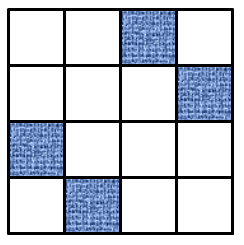}}\hspace{1cm}
	\subfloat[$\psi_2$]{\includegraphics[width = 2cm]{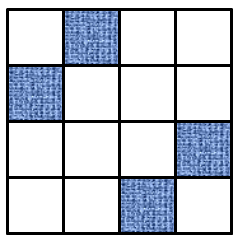}}\hspace{1cm}
	\subfloat[$\psi_3$]{\includegraphics[width = 2cm]{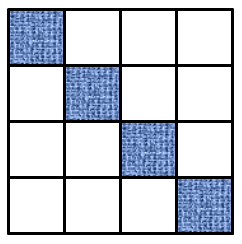}} 
	\caption{Alice's ontic states are shown along column (left to right) and Bob's are along row (down to up). Altogether there are $16$ ontic states.}
	\label{epis-bi}
\end{figure}

Measurements on the pairs of elementary systems in defined as partitioning of the
set of sixteen ontic states into disjoint epistemic states. A measurement defined in this way must be compatible with KBP. Every such compatible partitioning of the set of sixteen ontic states into four disjoint pure epistemic states yields a maximally informative measurement. One such measurement is $M\equiv\{S_I,S_{II},S_{III},S_{IV}\}$, where,
\begin{subequations}\label{toy-measurement}
	\begin{align}
	S_I&=(1.1)\vee(2.2)\vee(3.3)\vee(4.4),\quad S_{II}=(1.2)\vee(2.1)\vee(3.4)\vee(4.3),\\
	S_{III}&=(1.3)\vee(2.4)\vee(3.1)\vee(4.2),\quad S_{IV}=(1.4)\vee(2.3)\vee(3.2)\vee(4.1).
	\end{align}
\end{subequations}
This measurement is analogous to the Bell measurement in quantum mechanics and can be visualized pictorially as in Fig.\ref{spec-measurement}. 

\begin{figure}[h!]
	\begin{center} 
		\includegraphics[scale=0.5]{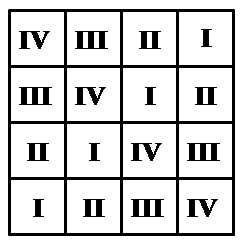}
	\end{center}
	\caption{Measurement $M$ as defined in Eq.(\ref{toy-measurement}).}
	\label{spec-measurement} 
\end{figure}  

{\bf Perfect toy-bit protocol for ALC task}\\
{\bf Encoding:} Alice and Bob start the protocol with the shared toy-bit entangled state $\psi_0$.  Consider the mapping $\{0,1\}^2\mapsto k$, with $k\in\{0,1,2,3\}$ as follows $00\mapsto 0,01\mapsto 1,10\mapsto 2,11\mapsto 3$. Whenever Alice (Bob) obtains a string $x$ ($y$) she (he) applies $U_k$ ($U_{k^\prime}$) defined in Eq.(\ref{toy-unitary}) on her (his) part of the entangled state $\psi_0$ and sends that part to Charlie. Straightforward calculation gives us,
\begin{subequations}
	\begin{align}
	U_0\circ U_0[\psi_0]=U_1\circ U_1[\psi_0]=U_2\circ U_2[\psi_0]=U_3\circ U_3[\psi_0]=\psi_0,\\
	U_0\circ U_1[\psi_0]=U_1\circ U_0[\psi_0]=U_2\circ U_3[\psi_0]=U_3\circ U_2[\psi_0]=\psi_1,\\
	U_0\circ U_2[\psi_0]=U_2\circ U_0[\psi_0]=U_1\circ U_3[\psi_0]=U_3\circ U_1[\psi_0]=\psi_2,\\
	U_0\circ U_3[\psi_0]=U_3\circ U_0[\psi_0]=U_1\circ U_2[\psi_0]=U_2\circ U_1[\psi_0]=\psi_3.
	\end{align}
\end{subequations}
The notation $U_k\circ U_{k'}[\psi]$ denotes that on Alice side permutation $U_{k}$ and on Bob side permutation $U_{k'}$ are applied.

{\bf Decoding:} For decoding, Charlie performs the measurement $M$ of Eq.(\ref{toy-measurement}). He answers $x=y$ while $S_{I}$ clicks, otherwise he answers $x\neq y$, resulting in perfect success.

\emph{Remark:} The above perfect toy-bit protocol for ALC task is different than the perfect quantum entangled strategy in a sense. While in quantum case Charlie performs a two outcome measurement for decoding and extracts only the information whether Alice's and Bob's strings are same or not, in the above case Charlie performs four-outcome measurement and hence extracts more information. However in toy-bit case also, Charlie can perform a two-outcome measurement $M'\equiv\{S_1,S_2\}$ (see Fig.\ref{spec-measurement1}) and can make the protocol exactly same as quantum perfect protocol. 
\begin{figure}[h!]
	\begin{center} 
		\includegraphics[scale=0.5]{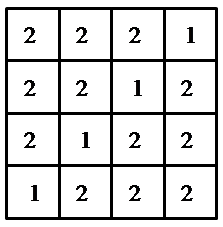}
	\end{center}
	\caption{Measurement $M'\equiv\{S_1,S_2\}$.}
	\label{spec-measurement1} 
\end{figure}  
\end{appendix}
\end{widetext}

\end{document}